\documentclass[11pt]{article}
\usepackage[utf8]{inputenc}
\usepackage{hyperref}
\usepackage{fullpage}
\usepackage{amssymb}
\usepackage{mathtools}
\usepackage{amsthm}
\usepackage{eucal}
\usepackage{dsfont}
\usepackage[T1]{fontenc}
\usepackage{lmodern}
\usepackage[stretch=10,shrink=10]{microtype}
\usepackage[capitalise]{cleveref}
\usepackage{color,soul}

\allowdisplaybreaks[1]

\theoremstyle{plain}
\newtheorem{theorem}{Theorem}[section]
\newtheorem{lemma}[theorem]{Lemma}

\newtheorem{cor}[theorem]{Corollary}

\theoremstyle{definition}
\newtheorem{definition}[theorem]{Definition}

\newtheorem{fact}[theorem]{Fact}

\newcommand {\eps} {\varepsilon}
\newcommand {\br} [1] {\ensuremath{ \left( #1 \right) }}
\newcommand {\Br} [1] {\ensuremath{ \left[ #1 \right] }}

\newcommand {\minusspace} {\: \! \!}

\newcommand {\fn} [2] {\ensuremath{ #1 \minusspace \br{ #2 } }}
\newcommand {\Fn} [2] {\ensuremath{ #1 \minusspace \Br{ #2 } }}

\newcommand {\defeq} {\ensuremath{ \stackrel{\mathrm{def}}{=} }}

\newcommand {\prob} [1] {\Fn{\Pr}{#1}}

\newcommand {\norm} [1] {\ensuremath{ \left\| #1 \right\| }}
\newcommand {\normsub} [2] {\ensuremath{ \norm{#1}_{#2} }}
\newcommand {\onenorm} [1] {\normsub{#1}{1}}

\DeclareMathOperator*{\bigE}{\mathbb{E}}
\newcommand {\expec} [2] {\Fn{\bigE_{\substack{#1}}}{#2}}
\newcommand {\email} [1] {\href{mailto:#1}{\texttt{#1}}}

\newcommand {\finset} [1] {\ensuremath{\CMcal{#1}}}
\newcommand {\bra} [1] {\ensuremath{ \left\langle #1 \right| }}
\newcommand {\ket} [1] {\ensuremath{ \left| #1 \right\rangle }}
\newcommand {\ketbratwo} [2] {\ensuremath{ \left| #1 \middle\rangle \middle\langle #2 \right| }}
\newcommand {\ketbra} [1] {\ketbratwo{#1}{#1}}
\newcommand {\cspace} [1] {\ensuremath{\mathnormal{#1}}}

\newcommand {\Tr} {\ensuremath{ \mathrm{Tr} }}
\newcommand {\partrace} [2] {\fn{\Tr_{#1}}{#2}}

\newcommand {\suppress}[1]{}

\newcommand {\set} [1] {\ensuremath{ \left\lbrace #1 \right\rbrace }}
\newcommand {\reg} [1] {\ensuremath{ \mathnormal{#1} }}

\def\C{\mathcal{C}}
\def\D{\CMcal{D}}

\def\R{\mathcal{R}}
\def\G{\CMcal{G}}
\def\P{\mathcal{P}}
\def\Q{\CMcal{Q}}

\def\T{\CMcal{T}}

\def\H{\mathcal{H}}

\def\F{\mathrm{F}}

\def\D{\mathrm{D}}
\def\E{\mathcal{E}}

\newcommand{\braket}[2]{\langle#1|#2\rangle}

\newcommand {\mytitle} {An upper bound on quantum capacity of unital channels}
\newcommand {\Anurag}  {Anurag Anshu}

\newcommand {\CQT} {Centre for Quantum Technologies}
\newcommand {\NUS} {National University of Singapore}

\newcommand {\authorblock} [3] {
	\begin{minipage}[t]{0.3\linewidth}
		\centering
		{#1}\\[0.8ex]
		{\footnotesize {#2}\\[-0.7ex]
		\email{#3}}
	\end{minipage}\vspace{1ex}
}

\hypersetup{
	pdfstartview={FitH},
	pdfdisplaydoctitle={true},
	breaklinks={true},
	bookmarksopen={true},
	bookmarksnumbered={false},
	pdftitle={\mytitle},
	pdfauthor={\Anurag}
}

\begin{document}

\title{\textbf{\mytitle}\\[2ex]}

\author{
    \authorblock{\Anurag}{\CQT, \NUS}{a0109169@u.nus.edu}
}

\maketitle

\begin{abstract}
We analyze the quantum capacity of a unital quantum channel, using ideas from the proof of near-optimality of Petz recovery map [Barnum and Knill 2000] and give an upper bound on the quantum capacity in terms of regularized output $2$-norm of the channel. We also show that any code attempting to exceed this upper bound must incur large error in decoding, which can be viewed as a weaker version of the strong converse results for quantum capacity. As an application, we find nearly matching upper and lower bounds (up to an additive constant) on the quantum capacity of quantum expander channels. Using these techniques, we further conclude that the `mixture of random unitaries' channels arising in the construction of quantum expanders in [Hastings 2007] show a trend in multiplicativity of output $2$-norm similar to that exhibited in [Montanaro 2013] for output $\infty$-norm of random quantum channels.
\end{abstract}

\section{Introduction}
One of the most fundamental developments in quantum information theory has been towards an understanding of various capacities of quantum channels. Quantum capacity of a quantum channel is characterized by a well known quantity called the coherent information (\cite{Nielsen96}). Similarly, classical capacity of a quantum channel is characterized by its Holevo information \cite{holevo73}. Unfortunately, a single letter formula for either the quantum capacity or the classical capacity is not known, and a regularization is needed to completely capture these capacities \cite{smith10}. 

 It was shown by Shor \cite{Shor2004} that the problem of regularization of Holevo information is related to various other additivity questions in quantum information and in particular to additivity of the minimum entropy output of a quantum channel. This was combined with an extensive study of multiplicativity of output norms of quantum channels (we discuss output $2$-norm and output $\infty$-norm in Section \ref{prelims}, general definition can be found in following references). Violations of multiplicativity of various output norms were shown in a series of results \cite{HolevoWerner02,HaydenWinter08,Cubitt08}, culminating in a proof of violation of additivity of minimum entropy output by Hastings \cite{Hastings09}. The work \cite{Montanaro13} studied the output $\infty$-norm of a random quantum channel, where it was shown that most quantum channels still satisfied a weaker version of the multiplicativity of output $\infty$-norm (Theorem $3$ in the reference \cite{Montanaro13}). 

In this work, we primarily consider the quantum capacities of unital channels and their output $2$-norms. We provide an upper bound on quantum capacity of such channels in terms of their regularized output $2$-norm. In addition, we prove a result that is reminiscent of the `strong converse theorems', which have received a great deal of attention in recent literature on quantum channel capacity (see for example, \cite{WarsiSharma13,Wilde2014,TomamichelWildeWinter14,WildeWinter2014,Gupta2015,MorganWinter14,Cooney2016} and references therein). 

\subsection*{Results and techniques}
We provide an upper bound on the quantum capacity and the zero error classical capacity of a quantum channel (Lemma \ref{dimofcodespace} for quantum capacity of a general channel, Corollary \ref{cor:unital} for quantum capacity of a unital channel and Lemma \ref{zeroerror} for zero error classical capacity of a general channel). Our bound is inspired from the near-optimality of Petz recovery map due to Barnum and Knill \cite{barnumknill}, which has been well studied in literature, such as for approximate quantum error correction\cite{huikhoon} and achievability results in quantum channel capacity \cite{Dutta15}. Using this bound, we derive an upper bound on quantum capacity of unital channels and also a weak form of strong converse theorem for quantum capacity:  for any encoding-decoding operation that attempts to exceed the upper bound on quantum capacity,  the success fidelity of decoding the quantum message falls exponentially in number of channel uses (Theorem \ref{main:theo}). 

As an application, we consider the well studied quantum expander channels (various constructions of which have been presented in \cite{Ambianis04, Has07, Har07,Gross07, HasHar08}), and in particular, the mixture of random unitaries as defined in \cite{Has07}. We find an upper and a lower bound on quantum capacities of such random channels, and show that with high probability, the upper and lower bounds differ by a small constant (Lemma \ref{lem:expanders} and Corollary \ref{cor:hastingsexpanders}). Moreover, along the lines of the result shown in \cite{Montanaro13}, we find that the output $2$-norm of such channels is nearly multiplicative (with high probability), with the multiplicativity exponent close to $1$ (Corollary \ref{cor:hastingsexpanders}).

\section{Preliminaries}
\label{prelims}
For integer $n \geq 1$, let $[n]$ represent the set $\{1,2, \ldots, n\}$. Let $\mathbb{R}$ represent the set of real numbers. We let $\log$ represent logarithm to the base $2$ and $\ln$ represent logarithm to the base $\mathrm{e}$. 
\suppress{Let $\finset{X}$ and $\finset{Y}$ be finite sets. $\finset{X}\times\finset{Y}$ represents the cross product of $\finset{X}$ and $\finset{Y}$. For a natural number $k$, we let 
$\finset{X}^k$ denote the set $\finset{X}\times\cdots\times\finset{X}$, the cross product of
$\finset{X}$, $k$ times.  Let $\mu$ be a probability distribution on $\finset{X}$. We let $\mu(x)$ represent the probability of $x\in\finset{X}$ according to $\mu$. We use the same symbol to represent a random variable and its distribution whenever it is clear from the context. The expectation value of function $f$ on $\finset{X}$ is defined as
$\expec{x \leftarrow X}{f(x)} \defeq \sum_{x \in \finset{X}} \prob{X=x}
\cdot f(x),$ where $x\leftarrow X$ means that $x$ is drawn according to distribution $X$.}

Consider a finite dimensional Hilbert space $\H$ endowed with an inner product $\langle \cdot, \cdot \rangle$ (In this paper, we only consider finite dimensional Hilbert-spaces). The $\ell_1$ norm of an operator $X$ on $\H$ is $\onenorm{X}\defeq\Tr\sqrt{X^{\dag}X}$ and $\ell_2$ norm is $\norm{X}_2\defeq\sqrt{\Tr XX^{\dag}}$. A quantum state (or a density matrix or a state) is a positive semi-definite matrix on $\H$ with trace equal to $1$. It is called {\em pure} if and only if its rank is $1$. A sub-normalized state is a positive semi-definite matrix on $\H$ with trace less than or equal to $1$. Let $\ket{\psi}$ be a unit vector on $\H$, that is $\langle \psi,\psi \rangle=1$.  With some abuse of notation, we use $\psi$ to represent the state and also the density matrix $\ketbra{\psi}$, associated with $\ket{\psi}$. Given a quantum state $\rho$ on $\H$, {\em support of $\rho$}, called $\text{supp}(\rho)$ is the subspace of $\H$ spanned by all eigen-vectors of $\rho$ with non-zero eigenvalues.
 
A {\em quantum register} $A$ is associated with some Hilbert space $\H_A$. Define $|A| \defeq \dim(\H_A)$. Let $\mathcal{L}(A)$ represent the set of all linear operators on $\H_A$. We denote by $\mathcal{D}(A)$, the set of quantum states on the Hilbert space $\H_A$. State $\rho$ with subscript $A$ indicates $\rho_A \in \mathcal{D}(A)$. If two registers $A,B$ are associated with the same Hilbert space, we shall represent the relation by $A\equiv B$.  Composition of two registers $A$ and $B$, denoted $AB$, is associated with Hilbert space $\H_A \otimes \H_B$.  For two quantum states $\rho\in \mathcal{D}(A)$ and $\sigma\in \mathcal{D}(B)$, $\rho\otimes\sigma \in \mathcal{D}(AB)$ represents the tensor product (Kronecker product) of $\rho$ and $\sigma$. The identity operator on $\H_A$ (and associated register $A$) is denoted $I_A$. 

Let $\rho_{AB} \in \mathcal{D}(AB)$. We define
\[ \rho_{\reg{B}} \defeq \partrace{\reg{A}}{\rho_{AB}}
\defeq \sum_i (\bra{i} \otimes I_{\cspace{B}})
\rho_{AB} (\ket{i} \otimes I_{\cspace{B}}) , \]
where $\set{\ket{i}}_i$ is an orthonormal basis for the Hilbert space $\H_A$.
The state $\rho_B\in \mathcal{D}(B)$ is referred to as the marginal state of $\rho_{AB}$. Unless otherwise stated, a missing register from subscript in a state will represent partial trace over that register. Given a $\rho_A\in\mathcal{D}(A)$, a {\em purification} of $\rho_A$ is a pure state $\rho_{AB}\in \mathcal{D}(AB)$ such that $\partrace{\reg{B}}{\rho_{AB}}=\rho_A$. Purification of a quantum state is not unique.

A quantum map $\E: A\rightarrow B$ is a completely positive linear map (mapping states in $\mathcal{D}(A)$ to states in $\mathcal{D}(B)$). In this work, we shall also consider maps that do not preserve trace. Trace preserving quantum maps shall be referred to as \textit{quantum channels}.  A {\em unitary} operator $U_A:\H_A \rightarrow \H_A$ is such that $U_A^{\dagger}U_A = U_A U_A^{\dagger} = I_A$. An {\em isometry}  $V:\H_A \rightarrow \H_B$ is such that $V^{\dagger}V = I_A$ and $VV^{\dagger} = I_B$. The set of all unitary operations on register $A$ is  denoted by $\mathcal{U}(A)$. A quantum channel $\E:A\rightarrow A$ is said to be \textit{unital} if it holds that $\E(I_A)=I_A$.

Given a quantum map $\E: A\rightarrow B$, maximum output $\infty$-norm of $\E$ is defined as $\|\E\|_{\infty}\defeq\text{max}_{\rho\in \mathcal{D}(A)}\{\|\E(\rho)\|\}$. Here, $\|.\|$ is the operator norm. We say that $\E$ obeys $\infty$-norm multiplicativity with exponent $\alpha$ if $\|\E^{\otimes n}\|_{\infty}\leq \|\E\|_{\infty}^{n\alpha}$. Similarly, maximum output $2$-norm of $\E$ is defined as $\|\E\|_2\defeq \text{max}_{\rho\in \mathcal{D}(A)}\{\Tr(\E^2(\rho))\}$. We say that $\E$ obeys $2$-norm multiplicativity with exponent $\alpha$ if $\|\E^{\otimes n}\|_{2}\leq \|\E\|_{2}^{n\alpha}$. 

Following fact says that the optimization in the definition of $\|\E\|_2$ is achieved by a pure state. 
\begin{fact}
\label{pureopt}
For every state $\rho$, there exists a pure state $\ket{\sigma}$ such that $\Tr(\E^2(\rho)) \leq \Tr(\E^2(\sigma))$.
\end{fact}
\begin{proof}
We consider the eigen-decomposition $\rho=\sum_i p_i \ketbra{\sigma_i}$. Then 
$$\Tr(\E^2(\rho)) = \sum_{i,j}p_ip_j\Tr(\E(\sigma_i)\E(\sigma_j)) \leq \sum_{i,j}p_ip_j\sqrt{\Tr(\E^2(\sigma_i))\Tr(\E^2(\sigma_j))} = (\sum_i p_i\sqrt{\Tr(\E^2(\sigma_i))})^2,$$ where we use the Cauchy-Schwartz inequality $\Tr(XY)\leq \sqrt{\Tr(X^2)\Tr(Y^2)}$ for hermitian matrices $X,Y$. Now, using concavity of square-root, we proceed as
$$\Tr(\E^2(\rho)) \leq \sum_{i}p_i\Tr(\E^2(\sigma_i)) \leq \mathrm{max}_i \Tr(\E^2(\sigma_i)).$$ Thus proves the fact.
\end{proof}

\subsection*{Quantum channel capacities}

Given a quantum channel $\E: A\rightarrow B$ that serves as noise, we shall be interested in two kinds of capacities: the quantum capacity and the zero error classical capacity. We first describe the quantum capacity. Fix an $n>0$ and consider a $d_C$ dimensional `source' Hilbert space $\H_S$ (the $C$ in the subscript stands for `codespace', as the dimension of the system is equal to the dimension of the codespace used to encode the quantum states in the system). An encoding operation maps the register $S$ onto registers $A_1,A_2,\ldots A_n$ as follows. Alice introduces an ancillary register $T$, in the state $\ketbra{0}_T$ and applies an isometry $ST\rightarrow A_1A_2\ldots A_nT$. Under this isometry, every vector $\ket{\psi_0}\in \H_S$ gets mappes to a vector in $\ket{\psi}\in \H_{A_1A_2\ldots A_nT}$, forming a subspace $C$ of dimension $d_C$. Final implementation of encoding map involves tracing out the register $T$ (which we represent by the map $\T \defeq \ketbra{0}_T\otimes\Tr_T$) and sending the registers $A_1,A_2,\ldots A_n$ sequentially through $\E$. Let the registers output on Bob's side be $B_1,B_2,\ldots B_n$. Bob then applies a decoding (or recovery) operation $\R_C$ to obtain the registers $A_1,A_2,\ldots A_n,T$. The aim is to recover the state $\psi$ with as high fidelity as possible. Note that this is equivalent to recovering $\psi_0$, as $\psi$ and $\psi_0$ are related to each other via an isometry. We shall consider the standard definition of fidelity: $\F(\rho,\sigma) \defeq \Tr(\sqrt{\sqrt{\rho}\sigma\sqrt{\rho}})$ for quantum states $\rho,\sigma$.

We say that $\R_C$ recovers with average fidelity $\eta$ if it holds that 
$$\int_{\psi\in C}\F(\psi,\R_C(\T\otimes \E^{\otimes n}(\psi))) d\psi = \eta,$$ where $d\psi$ is the Haar measure over the codespace $C$. We define $\Q^{\eta}_n(\E)$ as the largest possible value of $\log(d_C)$ such that there exists a register $T$, an encoding subspace $C$ and a decoding operation $\R_C$ such that average fidelity is $\eta$. Quantum capacity of channel $\E$, denoted as $\Q(\E)$ is then defined as 
$$\Q(\E)\defeq \text{lim}_{\eta\rightarrow 1}\text{lim}_{n\rightarrow \infty}\frac{1}{n}Q^{\eta}_n(\E).$$

Following well known result holds for $\Q(\E)$ (see \cite{wilde2012} for a detailed discussion)
\begin{fact}[The Lloyd-Shor-Devetak theorem][\cite{lloyd97,Shor02,Devetak05}] 
\label{LSD}
For a quantum channel $\E:A\rightarrow B$ introduce a reference register $R$ with dimension of $\H_R$ same as the dimension of $\H_A$. Then 
$$\Q(\E) \geq \text{max}_{\ket{\Psi_{RA}}\in \mathcal{D}(RA)}\left(S(\E(\Psi_A))-S(I_R\otimes \E(\Psi_{RA}))\right).$$ 
\end{fact}
\bigskip

The zero error classical capacity \cite{medeiros05} is defined as follows. Given a collection of $M$ messages $\{1,2,\ldots M\}$, Alice encodes each message $m$ into a quantum state $\rho_m \in \D(A_1A_2\ldots A_n)$ and sends the registers $A_1,A_2,\ldots A_n$ sequentially through the channel $\E$.  Receiving all the registers $B_1,B_2,\ldots B_n$, Bob applies a decoding operation $\R$ that recovers the message $m$ with zero error. We define $\C_n(\E)$ as the largest possible $\log(M)$ such that there exist quantum states $\{\rho_1,\rho_2,\ldots \rho_m\}$ and a recovery operation $\R$ such that $\R(\E(\rho_m)) = \ketbra{m}$. Zero error classical capacity of $\E$ is now defined as $\C(\E) \defeq \text{lim}_{n\rightarrow \infty}\frac{1}{n}\C_n(\E)$.

\section{Upper bound on capacities using Petz recovery map}
Given a noise $\E: X \rightarrow Y$ acting on certain register $X$, and any positive semi-definite operator $\Pi$ on register $Y$, we define the following associated map $\P_{\Pi}(\rho) \defeq \E^{\dagger}(\Pi^{-1}\rho\Pi^{-1})$. Here, the map $\E^{\dagger}: Y\rightarrow X$ is defined as $\Tr(\sigma\E^{\dagger}(\rho)) = \Tr(\E(\sigma)\rho)$ for all $\rho \in \mathcal{D}(Y)$ and $\sigma\in \mathcal{D}(X)$. The Petz recovery map is a special case when $\Pi$ is chosen to ensure that $\P_{\Pi}$ is a quantum channel. The following relation was essentially proved in \cite{barnumknill} (and elaborated in \cite{huikhoon}). We reproduce its proof in Appendix \ref{proof:petzoptimal}.

\begin{lemma}[\cite{barnumknill}, \cite{huikhoon}]
\label{barnumknill} For any quantum map $\R: Y\rightarrow X$, the noise $\E:X\rightarrow Y$, a positive semi-definite operator $\Pi$ on register $Y$ fully supported in the image of $\E$ and any state $\psi\in \H_X$, it holds that
$$\F^2(\psi,\R(\E(\psi))) \leq \sqrt{\bra{\psi}\P_{\Pi}(\E(\psi))\ket{\psi}\bra{\psi}\R(\Pi^2)\ket{\psi}}.$$ 
\end{lemma}

Now, as discussed in Section \ref{prelims}, consider the setting of $n$ registers $A_1,A_2\ldots A_n$, such that all $A_i\equiv A$. Let $\E:A\rightarrow B$ be a noise, which acts independently on above registers as $\E^{\otimes n}: A_1\otimes A_2\otimes \ldots A_n \rightarrow B_1\otimes B_2\otimes \ldots B_n$. For the operator $\T$, we consider the associated map $\P_{\sqrt{d_T}\ketbra{0}}$, which we simply abbreviate as $\P_T$.  Here $d_T$ is the dimension of $\H_T$. From the Kraus representation of $\T$ (that is, $\T(\rho) = \sum_i \ket{0}\bra{i}\rho\ket{i}\bra{0}$), it is easy to observe that $\P_T(\ketbra{0}) = \T^{\dagger}(\frac{\ketbra{0}}{d_T}) = \frac{I_T}{d_T}$.

For the channel $\E$ and operator $\Pi$ supported on the image of $\E^{\otimes n}$, define the following map: $$G_{\E,\Pi}(.)\defeq \Pi^{-1/2}\E(.)\Pi^{-1/2}.$$

 Then we have the following lemma.
\begin{lemma}
\label{dimofcodespace}
Given a noise $\E:A\rightarrow B$ such that dimension of $\H_A$ is $d$ and a codespace $C$ (along with register $T$ and recovery map $\R_C$) with average fidelity $\eta$, we have $$d_C\leq \frac{1}{\eta^4}\mathrm{min}_{\Pi} \|G_{\E^{\otimes n},\Pi}\|_2\cdot\Tr(\Pi^2),$$
where minimization is over all positive semi-definite operators $\Pi$ that are in the support of image of $\E^{\otimes n}$.  
\end{lemma}
\begin{proof}

Fix an orthonormal basis in $\H_{A_1A_2,\ldots A_nT}$ : $\{\ket{\phi_1},\ket{\phi_2}\ldots \ket{\phi_{d^n\cdot d_T}}\}$ such that for all $i\leq d_C$, $\phi_i\subset C$. Let $\Pi$ be any operator fully supported in the image of $\E^{\otimes n}$. Consider the following map associated to $\E^{\otimes n}$: $$\P_{\Pi}(.) \defeq \E^{\dagger\otimes n}(\Pi^{-1}(.)\Pi^{-1}).$$ 

We apply  Lemma \ref{barnumknill} to the `noise' $\T\otimes \E^{\otimes n}$ and the map $\P_T\otimes \P_{\Pi}$:
\begin{equation}
\label{petzonnoise}
\sum_i \F^4(\phi_i,\R_C(\T\otimes\E^{\otimes n}(\phi_i))) \leq \sum_i \bra{\phi_i}\P_T\otimes\P_{\Pi}(\T\otimes\E^{\otimes n}(\phi_i))\ket{\phi_i}\bra{\phi_i}\R_C(d_T\ketbra{0}\otimes \Pi^2))\ket{\phi_i}
\end{equation}

We shall upper bound each term $\bra{\phi_i}\P_T\otimes\P_{\Pi}(\T\otimes\E^{\otimes n}(\phi_i))\ket{\phi_i}$ as follows.
\begin{align*}
&\bra{\phi_i}\P_T\otimes\P_{\Pi}(\T\otimes\E^{\otimes n}(\phi_i))\ket{\phi_i}\\
&= \bra{\phi_i}\P_T(\ketbra{0})\otimes\P_{\Pi}(\E^{\otimes n}(\Tr_T\phi_i))\ket{\phi_i}\\ 
&(\text{as } \T \mbox{ traces out register } T \mbox{ and replaces it with the state } \ketbra{0})\\ 
&= \bra{\phi_i}\frac{I_T}{d_T}\otimes\P_{\Pi}(\E^{\otimes n}(\Tr_T\phi_i))\ket{\phi_i} \\ 
&(\text{as } \P_T \mbox{ replaces the state }\ketbra{0} \mbox{ with the maximally mixed state on register } T)\\
&\leq \frac{1}{d_T}\text{max}_i\Tr((\Tr_T\phi_i)\P_{\Pi}(\E^{\otimes n}(\Tr_T\phi_i)))\\ 
&= \frac{1}{d_T}\text{max}_i\Tr(\E^{\otimes n}(\Tr_T\phi_i)\Pi^{-1}\E^{\otimes n}(\Tr_T\phi_i)\Pi^{-1})\\ 
& (\mbox{follows by incorporating the definition of the map } \P_{\Pi}) \\ 
&= \frac{1}{d_T}\text{max}_i\Tr(\Pi^{-1/2}\E^{\otimes n}(\Tr_T(\phi_i))\Pi^{-1/2}\Pi^{-1/2}\E^{\otimes n}(\Tr_T(\phi_i))\Pi^{-1/2}) \\ 
& (\mbox{writing } \Pi^{-1} =\Pi^{-1/2}\Pi^{-1/2}\mbox{ and then using cyclicity of trace})\\
&= \frac{1}{d_T}\text{max}_i\Tr((G_{\E^{\otimes n},\Pi}(\Tr_T(\phi_i)))^2) \leq \frac{1}{d_T}\|G_{\E^{\otimes n},\Pi}\|_2.
\end{align*}

Applying it in Equation \ref{petzonnoise}, we obtain

\begin{align*}
&\sum_i \F^4(\phi_i,\R_C(\T\otimes\E^{\otimes n}(\phi_i)))\\ 
& \leq \frac{\|G_{\E^{\otimes n},\Pi}\|_2}{d_T}\sum_i \bra{\phi_i}\R_C(d_T\ketbra{0}\otimes \Pi^2))\ket{\phi_i}\\
& (\mbox{as each term } \bra{\phi_i}\R_C(d_T\ketbra{0}\otimes \Pi^2))\ket{\phi_i} \mbox{ is positive}) \\
 &= \|G_{\E^{\otimes n},\Pi}\|_2\sum_i \Tr(\R_C(\ketbra{0}\otimes \Pi^2)\phi_i) \\ 
&= \|G_{\E^{\otimes n},\Pi}\|_2\cdot \Tr(\R_C(\ketbra{0}\otimes \Pi^2)\sum_i \phi_i) \\ 
&= \|G_{\E^{\otimes n},\Pi}\|_2\cdot \Tr(\R_C(\ketbra{0}\otimes \Pi^2)) \\ 
& (\sum_i \phi_i = I^{\otimes n}\otimes I_T, \mbox{ since } \phi_i \mbox{ form an orthonormal basis})\\ 
&=   \|G_{\E^{\otimes n},\Pi}\|_2\cdot \Tr(\Pi^2)\quad  (\mbox{as the map } \R_C \mbox{ is a trace preserving quantum map})
\end{align*}

On the other hand, $$\sum_i \F^4(\phi_i,R_C(\T\otimes \E^{\otimes n}(\phi_i)))\geq d_C\cdot \frac{\sum_{i\leq d_C} \F^4(\phi_i,\R_C(\T\otimes \E^{\otimes n}(\phi_i)))}{d_C}.$$
Thus, we obtain
$$ \|G_{\E^{\otimes n},\Pi}\|_2\cdot \Tr(\Pi^2) \geq d_C\cdot \frac{\sum_{i\leq d_C} \F^4(\phi_i,\R_C(\T\otimes \E^{\otimes n}(\phi_i)))}{d_C}.$$
Now, averaging over all possible basis in codespace $C$, we find that 
$$ \|G_{\E^{\otimes n},\Pi}\|_2\cdot \Tr(\Pi^2) \geq d_C\cdot \int_{\phi\in C} \F^4(\phi_i,\R_C(\T\otimes \E^{\otimes n}(\phi)))d\phi \geq d_C\cdot (\int_{\phi\in C} \F(\phi_i,\R_C(\T\otimes \E^{\otimes n}(\phi)))d\phi)^4,$$ where last inequality follows by convexity of the function $x\rightarrow x^4$. This proves the lemma, by incorporating the definition of average fidelity $\eta$ and optimizing over all possible positive semi-definite operators $\Pi$ supported in the image of $\E^{\otimes n}$.
\end{proof} 

We have the following corollary of above lemma, which gives an upper bound on the quantum capacity and also says that exceeding this upper bound leads to decrease in average fidelity exponentially in $n$. Since we shall use this corollary in later sections for unital channels, we have restricted its statement to such channels. 

\begin{cor}
\label{cor:unital}
Suppose the channel $\E$ is unital. Then we have that 
$$\Q(\E)\leq \log(d\cdot \mathrm{lim}_{n\rightarrow \infty}\|\E^{\otimes n}\|^{1/n}_2).$$
Furthermore, let $C$ be any codespace of dimension $d_C = d^n\|\E^{\otimes n}\|_2(1+\beta)^n$, for some $\beta>0$. Then the average fidelity $\eta$ satisfies the following relation, irrespective of the recovery map: $$\eta^4 \leq \frac{1}{(1+\beta)^n}.$$
\end{cor}
\begin{proof}
In Lemma \ref{dimofcodespace}, we choose $\Pi = I^{\otimes n}$. This gives $\G_{\E^{\otimes n},\Pi} = \E^{\otimes n}$ and we find that 
$$\frac{1}{n}\log(d_C)\leq \log(d\cdot \mathrm{lim}_{n\rightarrow \infty}\|\E^{\otimes n}\|^{1/n}_2/\eta^4).$$ Now we take the limit $n\rightarrow \infty$ and then take $\eta\rightarrow 1$. Second part of the corollary proceeds by direct substitution in Lemma \ref{dimofcodespace}, with the choice of $\Pi = I^{\otimes n}$.
\end{proof}

For the zero error classical capacity of $\E$, similar result is shown to hold. 
\begin{lemma}
\label{zeroerror}
It holds that $\C(\E) \leq \mathrm{lim}_{n\rightarrow \infty}\frac{1}{n}\log(\mathrm{min}_{\Pi}\|G_{\E^{\otimes n},\Pi}\|_2\cdot \Tr(\Pi^2))$. 
\end{lemma}
\begin{proof}
Given the constraint $\R_C(\E(\rho_m))=\ketbra{m}$, we find that $\rho_m\rho_{m'}=0$ if $m\neq m'$. Now, for the mapping $m\rightarrow \rho_m$, we consider a purifying register $T$ such that $\psi_m \in \D(A_1A_2\ldots A_nT)$ is a purification of $\rho_m$. Clearly, $\psi_1,\psi_2,\ldots\psi_M$ form a basis in a $M$-dimensional subspace of $\H_{A_1A_2\ldots A_nT}$. Thus, we can repeat the analysis in Lemma \ref{dimofcodespace} with $\eta=1$, from which this lemma follows.  
\end{proof}

\section{Regularized $2$-norm for unital channels and capacity of expanders}

In this section, we shall restrict ourselves to unital channels acting on a $d$ dimensional Hilbert space. Let the Kraus decomposition of $\E: A\rightarrow A$ be $\E(.) = \sum_i E_i (.) E_i^{\dagger}$. Since $\E$ is unital, $I$ is a fixed point of $\E$ with eigenvalue $1$. Second largest singular value of $\E$ is defined as $\lambda_2(\E) \defeq \mathrm{max}_{\rho: \Tr(\rho) = 0, \Tr(\rho^{\dagger}\rho)=1} \sqrt{\Tr(\E(\rho)^{\dagger}\E(\rho))}$. Then we have the following lemma, proved in Appendix \ref{2-normmultiplicativity}.
\begin{lemma}
\label{lem:multiplicativebound}
Let $\E$ be a unital channel with second largest singular value $\lambda_2 \defeq \lambda_2(\E)<1$. For all $n\geq 1$, it holds that $$\|\E^{\otimes n}\|_2 \leq (\frac{1}{d}+\lambda_2^2)^n.$$ In particular, 
$\mathrm{lim}_{n\rightarrow \infty}\|\E^{\otimes n}\|^{1/n}_2 \leq (\frac{1}{d}+\lambda_2^2)$.
\end{lemma}

Combining with Corollary \ref{cor:unital}, we obtain our main theorem in a straightforward manner.
\begin{theorem}
\label{main:theo}
Let $\E$ be a unital channel. Then we have that 
$$\Q(\E)\leq \log(1+d\cdot\lambda_2^2).$$
Furthermore, let $C$ be any codespace of dimension $d_C = (1+d\lambda_2^2)^n(1+\beta)^n$, for some $\beta>0$. Then the average fidelity $\eta$ satisfies the following relation, irrespective of the recovery map: $$\eta^4 \leq \frac{1}{(1+\beta)^n}.$$
\end{theorem}

\subsection{Expander channels}

\begin{definition}
\label{def:expander}
A unital quantum channel $\E:A\rightarrow A$ with $k$ Kraus operators $\{E_i\}_{i=1}^k$ and acting on a $d$-dimensional Hilbert space $\H_A$ is said to be a $(C,k,d)$-expander if it holds that $\lambda^2_2(\E) = \frac{C}{k}$.  
\end{definition}

Under this definition, we obtain the following Lemma.
\begin{lemma}
\label{lem:expanders}
Given a channel $\E: A\rightarrow A$ that is a $(C,k,d)$-expander. Then following properties hold for $\E$. 
\begin{itemize}
\item The quantum capacity $\Q(\E)$ is upper bounded by $\log(d)-\log(k) + \log(C + \frac{k}{d})$ and lower bounded by $\log(d)-\log(k)$. 
\item If $\frac{\log(d_C)}{n} = \log(d)-\log(k) + \log((C+\frac{k}{d})(1 + \beta))$, then average fidelity $\eta$ satisfies $\eta^4 < (1+\beta)^{-n}.$
\item  For all $n$, it holds that $\|\E^{\otimes n}\|_2^{1/n} \leq \frac{1}{d} + \frac{C}{k}$ and $\|\E\|_2 \geq \frac{1}{k}$.
\end{itemize}
\end{lemma}
\begin{proof}
We prove each item separately.
\begin{itemize}
\item The upper bound on $\Q(\E)$ follows from Theorem \ref{main:theo} and the assumption in Definition \ref{def:expander} that $\lambda^2_2(\E)= \frac{C}{k}$. For the lower bound, we recall from the Lloyd-Shor-Devetak theorem (Fact \ref{LSD}) that $\Q(\E) \geq \text{max}_{\ket{\Psi_{RA}}}(S(\E(\Psi_A))-S(I_R\otimes \E(\Psi_{RA}))).$ Now let $\Psi_A \defeq \frac{I}{d}$. Then $S(\E(\Psi_A)) = S(\Psi_A) = \log(d)$. On the other hand, $S(I_R\otimes\E(\Psi_{RA})) \leq \log(k)$ as $\Psi_{RA}$ is a pure state and $\E$ is composed of $k$ Kraus operators (which means that $\E(\Psi_{RA})$ is a convex combination of $k$ pure states). Hence $\Q(\E) \geq \log(d) -\log(k)$.

\item Second item again follows from Theorem \ref{main:theo}. 
\item For the third item, we observe that $\|\E\|_2 \geq \frac{1}{k}$ for any channel. This follows because $\|\E\|_2 = \mathrm{max}_{\ket{\psi}}\Tr(\E(\psi)^2)$. Now, $\E(\psi)$ is a convex combination of $k$ pure states and hence $\Tr(\E(\psi)^2) > \frac{1}{k}$. This proves the item when combined with Lemma \ref{lem:multiplicativebound}.
\end{itemize}
\end{proof}

A well known example of expander construction is due to Hastings \cite{Has07}, who showed the following theorem. 

\begin{theorem}[\cite{Has07}]
\label{hastingstheorem}
Pick $k/2$ unitary operators $\{U_1,U_2,\ldots U_{k/2}\}$ (each acting on $d$ dimensional Hilbert space) from the Haar measure and construct the quantum channel $\E(\rho) \defeq \frac{1}{k} \sum_i (U_i\rho U_i^{\dagger} + U_i^{\dagger}\rho U_i)$. Then for every $\eps>0$, with probability at least $1- e^{- \eps \cdot d^{2/15}}$, $\E$ is a $(4+4\eps, k, d)$-expander.
\end{theorem}

Combining this with Lemma \ref{lem:expanders}, we obtain the following straightforward corollary. The third item below is similar in spirit to the result in \cite{Montanaro13}.

\begin{cor}
\label{cor:hastingsexpanders}
Consider a random channel $\E$ as constructed in Theorem \ref{hastingstheorem}. Then for every $\eps>0$, setting $d> \frac{k}{\eps}$, the following holds with probability at least $1-e^{-\eps \cdot d^{2/15}}$ . 
\begin{itemize}
\item The quantum capacity $\Q(\E)$ is upper bounded by $\log(d)- \log(k) + \log(4+5\eps)$ and lower bounded by $\log(d)-\log(k)$. 
\item If $\frac{\log(d_C)}{n} = \log(d)-\log(k)+\log((4+5\eps)(1 + \beta))$, then average fidelity $\eta$ decays as $\eta^4 < (1+\beta)^{-n}.$  
\item $\|\E^{\otimes n}\|_2 \leq \|\E\|_2^{n(1 + \frac{4}{\log(k)} )}$.
\end{itemize}
\end{cor}

\section{Acknowledgements}
I would like to thank Aram Harrow, Ashley Montanaro, Andreas Winter, Debbie Leung, Jamie Sikora and Mark M. Wilde for helpful discussions. I am grateful to Andreas Winter for pointing out that Lemma \ref{zeroerror} follows from the arguments presented in Lemma \ref{dimofcodespace}. This work is supported by the Core Grants of Centre for Quantum Technologies, the National Research Foundation, Prime Minister's Office, Singapore and the
Ministry of Education, Singapore under the Research Centres of Excellence programme. 

\bibliographystyle{alpha}
\bibliography{references}

\appendix

\section{Proof of Lemma \ref{barnumknill}}
\label{proof:petzoptimal}
\begin{proof}
Let $\{R_i\}, \{E_i\}$ be respective Kraus operators for $\R$ and $\E$. That is, $\R(\rho) = \sum_k R_k\rho R^{\dagger}_k$ and similarly for $\E$. Then we have $\F^2(\psi,\R(\E(\psi))) = \sum_{i,j}|\bra{\psi}R_jE_i\ket{\psi}|^2$. Consider the matrix $X_{ij}\defeq \bra{\psi}R_jE_i\ket{\psi}$. By singular-value decomposition, there 
 exist unitaries $U,V$ with respective entries $\{u_{k,i}\}_{k,i}, \{v_{l,j}\}_{l,j}$ such that $Y_{k,l} \defeq \sum_{i,j}X_{ij}u_{k,i}v_{l,j}$ is a diagonal matrix and $\sum_k |Y_k|^2 = \sum_{i,j}|X_{ij}|^2$. Let $E'_k\defeq \sum_i u_{k,i}E_i$ and $R'_l\defeq \sum_j v_{l,j}R_j$ be new Kraus operators for $\R$ and $\E$ respectively. Then we have that 
\begin{align*}
\F^2(\psi,\R(\E(\psi))) &= \sum_k |Y_k|^2 = \sum_{k}|\bra{\psi}R'_kE'_k\ket{\psi}|^2\\
 & = \sum_k|\Tr(\psi R'_kE'_k\psi)|^2 = \sum_k|\Tr(\psi R'_k\Pi^{1/2}\Pi^{-1/2}E'_k\psi)|^2 \\ 
& (\mbox{as }\Pi \mbox{ is fully supported in the image of }\E)\\ 
& \leq \sum_k\Tr(\psi R'_k\Pi R'^{\dagger}_k)\Tr(\Pi^{-1}E'_k\psi E'^{\dagger}_k) \\ 
& (\mbox{using Cauchy-Schwartz inequality})\\ 
& \leq \sqrt{\sum_k|\bra{\psi}R'_k\Pi R'^{\dagger}_k\ket{\psi}|^2}\sqrt{\sum_k |\bra{\psi}E'^{\dagger}_k\Pi^{-1}E'_k\ket{\psi}|^2}\\ 
& \leq \sqrt{\sum_k\bra{\psi}R'_k\Pi^2R'^{\dagger}_k\ket{\psi}}\sqrt{\sum_k |\bra{\psi}E'^{\dagger}_k\Pi^{-1}E'_k\ket{\psi}|^2}\\
& =\sqrt{\bra{\psi}\R(\Pi^2)\ket{\psi}}\sqrt{\sum_k |\bra{\psi}E'^{\dagger}_k\Pi^{-1}E'_k\ket{\psi}|^2}\\
& \leq \sqrt{\bra{\psi}\R(\Pi^2)\ket{\psi}}\sqrt{\sum_{k,l} |\bra{\psi}E'^{\dagger}_k\Pi^{-1}E'_l\ket{\psi}|^2} =\sqrt{\bra{\psi}\R(\Pi^2)\ket{\psi}}\sqrt{\bra{\psi}\P_{\Pi}\E(\psi)\ket{\psi}}
\end{align*}

We explain the second last inequality, which says that $\sum_k|\bra{\psi}R'_k\Pi R'^{\dagger}_k\ket{\psi}|^2 \leq \sum_k\bra{\psi}R'_k\Pi^2R'^{\dagger}_k\ket{\psi}$. Consider
$$\sum_k|\bra{\psi}R'_k\Pi R'^{\dagger}_k\ket{\psi}|^2 \leq \sum_{k,l}|\bra{\psi}R'_k\Pi R'^{\dagger}_l\ket{\psi}|^2 = \sum_{k,l}\bra{\psi}R'_k\Pi R'^{\dagger}_l\ket{\psi}\bra{\psi}R'_l\Pi R'^{\dagger}_k\ket{\psi}.$$ Let $\ket{\phi_k} \defeq \Pi R'^{\dagger}_k\ket{\psi}$. Observe that $\sum_l R'^{\dagger}_l\ketbra{\psi}R'_l < \sum_l R'^{\dagger}_lR_l = I$, since $I-\ketbra{\psi}$ is a positive semidefinite operator and hence $R'^{\dagger}_l(I-\ketbra{\psi})R'_l$ is a positive semidefinite operator. This implies 
$$\sum_k|\bra{\psi}R'_k\Pi R'^{\dagger}_k\ket{\psi}|^2 \leq \sum_{k}\bra{\phi_k}(\sum_l R'^{\dagger}_l\ket{\psi}\bra{\psi}R'_l)\ket{\phi_k} \leq \sum_k \bra{\phi_k}\ket{\phi_k} = \sum_k \bra{\psi}R'_k\Pi^2R'^{\dagger}_k\ket{\psi}.$$

This completes the proof.
\end{proof}

\section{Proof of Lemma \ref{lem:multiplicativebound}}
\label{2-normmultiplicativity}
\begin{proof}
We consider the mapping $$\ket{i}\bra{j} \rightarrow \ket{i}\ket{j}.$$ Under this mapping, a matrix $A=\sum_{i,j}a_{ij}\ket{i}\bra{j}$ goes to a `vector' $\ket{A} \defeq a_{ij}\ket{i}\ket{j}$ and a rank-$1$ state $\ketbra{\phi}$ goes to $\ket{\phi}\ket{\phi^{*}}$.  The inner product becomes $\braket{B}{A} =\sum_{ij} b^{*}_{ij}a_{ij} = \Tr(B^{\dag}A)$ which is the usual Hilbert-Schmidt inner product. The channel $\E$ gets mapped to the matrix $E = \sum_i E_i \otimes E^{*}_i$. 

The fact that $I$ is a fixed point of $\E$ implies that for $\ket{I}=\sum_{i=1}^d \ket{i}\ket{i}$, we have $E\ket{I}=\ket{I}$. Second largest singular value of $E$ (which we call $\lambda_2$) is the second largest eigenvalue of $\sqrt{E^{\dagger}E}$. Let $P^0\defeq \frac{1}{d}\ketbra{I}$ be projector onto the vector $\ketbra{I}$ (it is easy to check that $P^0\cdot P^0 = \frac{1}{d^2}\ketbra{I}\ketbra{I} = \frac{1}{d}\ketbra{I}$) and $P^1 \defeq \mathbb{I}- P^0$ be projector onto subspace orthogonal to $\ketbra{I}$. Then we have the following relations 
\begin{equation}
\label{matrix_inequalities}
P^1E^{\dagger}EP^1 < \lambda_2^2 P^1, P^0E^{\dagger}EP^0=P^0, P^1E^{\dagger}EP^0= 0.
\end{equation}

Consider the quantity $\|\E^{\otimes n}\|_2$ and recall that the optimisation in its definition is achieved by a pure state (Fact \ref{pureopt}). Let the optimal pure state be $\ket{\phi}$. We note that the state $\E^{\otimes n}(\ketbra{\phi})$ gets mapped to the vector $E^{\otimes n}\ket{\phi}\ket{\phi^*}$. 

Thus, we have $\|\E^{\otimes n}\|_2 = \text{max}_{\phi}\bra{\phi}\bra{\phi^*}(E^{\dag}E)^{\otimes n}\ket{\phi}\ket{\phi^*}$. For a string $s\in \{0,1\}^n$, define $P^s \defeq P^{s_0}\otimes P^{s_1}\otimes \ldots P^{s_n}$. This implies 
\begin{eqnarray*}
\bra{\phi}\bra{\phi^*}(E^{\dag}E)^{\otimes n}\ket{\phi}\ket{\phi^*}&=& \sum_{s,s'\in \{0,1\}^n} \bra{\phi}\bra{\phi^*}P^{s'}(E^{\dag}E)^{\otimes n}P^{s'}\ket{\phi}\ket{\phi^*}\quad \text{(Resolution of Identity)}\\ &=& \sum_{s\in \{0,1\}^n} \bra{\phi}\bra{\phi^*}P^{s}(E^{\dag}E)^{\otimes n}P^s\ket{\phi}\ket{\phi^*}\quad (\text{as } P^1E^{\dagger}EP^0= 0)\\ &\leq & \sum_{s\in\{0,1\}^n}\lambda_2^{2|s|} \bra{\phi}\bra{\phi^*}P^{s}\ket{\phi}\ket{\phi^*}\quad (\text{Equation  \ref{matrix_inequalities}})\\&\leq & \sum_{s\in\{0,1\}^n}\lambda_2^{2|s|}\frac{1}{d^{n-|s|}} \bra{\phi}\bra{\phi^*}\otimes_{i:s_i=0}\ketbra{I}_i\ket{\phi}\ket{\phi^*} \quad (P^1 < \mathbb{I} \text{ , } P^0 = \frac{1}{d}\ketbra{I}) \\&=& \sum_{s\in\{0,1\}^n}\lambda_2^{2|s|}\frac{1}{d^{n-|s|}} \Tr(\otimes_{i:s_i=0}\ketbra{I}_i\cdot\ketbra{\phi}\otimes\ketbra{\phi^*}) 
\end{eqnarray*}

 Now fix an $s$ and let $J_s$ be set of qudits on which $s_i=0$. Let $\bar{J}_s$ be rest of the qudits. Let $\ket{I}_{J_s}\defeq \otimes_{i\in J_s}\ket{I}_i = \sum_{t \in \{1,2\ldots d\}^{|J_s|}}\ket{t}_{J_s}\ket{t}_{J_s}$ be the maximally entangled unnormalized state on qudits in $J_s$.  Let $\rho= \Tr_{\bar{J}_s}\ketbra{\phi}$. Then 
$$\Tr(\otimes_{i:s_i=0}\ketbra{I}_i\cdot\ketbra{\phi}\otimes\ketbra{\phi^{*}}) = \bra{I}_{J_s}\rho\otimes \rho^{*} \ket{I}_{J_s} = \sum_{t,t'}\bra{t}\rho\ket{t'}\bra{t}\rho^{*}\ket{t'}$$$$=\sum_{t,t'}\bra{t}\rho\ket{t'}\bra{t'}\rho\ket{t} = \Tr(\rho^2) < 1.$$
The last inequality follows since $\rho$ is a quantum state.

This gives $$\|\E^{\otimes n}\|_2 \leq  \sum_{s\in\{0,1\}^n}\lambda_2^{2|s|}\frac{1}{d^{n-|s|}} = \sum_{|s|=0}^n{n \choose |s|}\lambda_2^{2|s|}\frac{1}{d^{n-|s|}} = (\frac{1}{d}+\lambda_2^2)^n.$$

\end{proof}

\end{document}